\documentclass[onecolumn]{IEEEtran}
\usepackage{gensymb}
\usepackage{graphicx}
\usepackage{bm,epstopdf}
\usepackage{soul,multirow}

\usepackage{mathrsfs}

\usepackage{amsmath,amssymb,amsthm,wasysym,latexsym,color}
\usepackage{url,algpseudocode} 
\usepackage[boxruled]{algorithm2e}
\usepackage[font=footnotesize,center]{caption}
\makeatletter
\renewcommand{\algocf@caption@boxruled}{%
	\vspace{-1pt}
	\hbox{%
		\vbox{
			\vskip\interspacetitleboxruled%
			\unhbox\algocf@capbox\hfill
			\vskip\interspacetitleboxruled
		}%
		\hskip-4pt%
	}\nointerlineskip%
}%
\makeatother
\makeatletter
\renewcommand{\@algocf@capt@plain}{above}
\makeatother

\SetAlCapSkip{1em}

\SetKwInput{KwInput}{Input}
\SetKwInput{KwOutput}{Output}

\newtheorem*{theorem*}{Theorem}

\usepackage{cite}

\definecolor{orange}{RGB}{255,107,0}

\definecolor{shadecolor}{RGB}{220,220,220}

\theoremstyle{remark}
\theoremstyle{remark*}\newtheorem*{remark*}{Remark}

\newtheorem{claim}{Claim}
\newtheorem{Lemma}{Lemma}
\theoremstyle{proposition}\newtheorem{proposition}{Proposition}
\theoremstyle{assumption}

\newcommand{\ahmed}[1]{{\color{black} {#1}}}
\begin{document}
	
	\title{Data-Driven Learning-Based Optimization for Distribution System State Estimation}
	
	\author{
	 Ahmed S. Zamzam, \IEEEmembership{Student Member,~IEEE,}
		Xiao Fu, \IEEEmembership{Member,~IEEE,} and
		Nicholas D. Sidiropoulos,~\IEEEmembership{Fellow,~IEEE}
		\thanks{The work of A.S. Zamzam and N.D. Sidiropoulos was supported in part by NSF grant 1525194. A. S. Zamzam is with the ECE Dept., Univ. of Minnesota, Minneapolis, MN 55455, USA. X. Fu is with the School of Electrical Engineering and Computer Science, Oregon State University, Corvallis, OR, 97331. N. D. Sidiropoulos is with the Department of Electrical and Computer Engineering, University of Virginia, Charlottesville, VA 22904.
			E-mails: ahmedz@umn.edu, xiao.fu@oregonstate.edu, nikos@virginia.edu.
			}
	}
	\maketitle
\begin{abstract}
Distribution system state estimation (DSSE) is a core task for monitoring and control of distribution networks. Widely used algorithms such as Gauss-Newton perform poorly with the limited number of measurements typically available for DSSE, often require many iterations to obtain reasonable results, and sometimes fail to converge. DSSE is a non-convex problem, and working with a limited number of measurements further {aggravates} the situation, as indeterminacy induces multiple global (in addition to local) minima. Gauss-Newton is {also} known to be sensitive to initialization{. Hence,} the situation is far from ideal. It {is} therefore natural to ask if there is a smart way of initializing Gauss-Newton that will avoid these DSSE-specific pitfalls. This paper proposes using historical or simulation-derived data to train a shallow neural network to `learn to initialize' -- that is, map the available measurements to  a point in the neighborhood of the true latent states (network voltages), which is used to initialize Gauss-Newton. It is shown that this hybrid machine learning / optimization approach yields superior performance in terms of stability, accuracy, and runtime efficiency, compared to conventional optimization-only approaches. It is also shown that judicious design of the neural network training cost function helps to improve the overall DSSE performance. 
\end{abstract}

\begin{keywords}
Distribution network state estimation, phasor measurement units, machine learning, neural networks, Gauss-Newton, least squares approximation.
\end{keywords}

\section{Introduction}
\label{sec:intro}
State estimation (SE) techniques are used to monitor power grid operations in real-time. Accurately monitoring the network operating point is critical for many control and automation tasks, such as Volt/VAr optimization, feeder reconfiguration and restoration. SE uses measured quantities like nodal voltages, injections, and line flows, together with physical laws in order to obtain an estimate of the system state variables, i.e., bus voltage magnitudes and angles~\cite{kekatos2017psse, dsse2018wgcs} across the network. SE techniques have  also proven to be useful in network `forensics', such as spotting bad measurements and identifying gross modelling errors~\cite{lin2018highly}.
	
Unlike transmission networks where measurement units are placed at almost all network nodes, the SE task in distribution systems is particularly challenging due to the scarcity of real-time measurements. This is usually compensated by the use of so-called \emph{pseudo-measurements}. Obtained through short-term load and renewable energy forecasting techniques, these pseudo-measurements play a vital rule in enabling distribution system state estimation (DSSE)~\cite{ghosh1997load, manitsas2012distribution, dvzafic2017multi}.
Several DSSE solvers based on weighted least squares (WLS) transmission system state estimation methods have been proposed~\cite{baran1994, li1996state, singh2009choice, kekatos2013distributed,  Wang2018}. A three-phase nodal voltage formulation was used to develop a WLS-based DSSE solver in~\cite{baran1994, li1996state}. Recently, the authors of~\cite{dzafic2018hybrid} used Wirtinger calculus to devise a new approach for WLS state estimation in the complex domain. In order to reduce the computational complexity and storage requirements, the branch-based WLS model was proposed in~\cite{baran1995branch, wang2004revised}. However, such gains can be only obtained when the target power system features only wye-connected loads that are solidly grounded. 
It is also recognized that incorporating phasor measurements in DSSE improves the observability and the estimation accuracy~\cite{phadke1986state, zivanovic1996implementation}. Therefore, the DSSE approach developed in this paper considers the case where classical (quadratic) and phasor (linear) measurements are available, as well as pseudo-measurements provided through short term forecasting algorithms.

WLS DSSE is a non-convex problem {that may have multiple} local minima, and working with a limited number of measurements {empirically aggravates} \textcolor{black} the situation, as it may introduce multiple local minima as well. Furthermore, Gauss-Newton type algorithms {behave very differently when using different initializations---the algorithms} may need many iterations, or even fail to converge. It is therefore natural to ask if there is a smart way of initializing Gauss-Newton that will avoid these pitfalls? 

\noindent {\bf Contributions.} 
In this paper, we propose a novel learning architecture for the DSSE task. Our idea is as follows. A wealth of historical data is often available for a given distribution system. This data is usually stored and utilized in various other network management tasks, such as load and injection forecasting. Even without detailed recording of the network state, we can reuse this data to simulate network operations off-line. We can then think of network states and measurements as {\sf (output,input)} training pairs, which can be used to train a neural network (NN) to `learn' a function that maps measurements to states. After the mapping function is learned, estimating the states associated with a fresh set of measurements only requires very simple operations---passing the measuremnts through the learned NN. This would greatly improve the efficiency of DSSE, bringing real-time state estimation within reach. 
Accurate and cheap DSSE using an NN may sound too good to be true, and in some sense (in its raw form) it is; but there is also silver lining, as we will see. 

Known as universal function approximators, neural networks have made a remarkable comeback in recent years, outperforming far more complicated (and disciplined) methods in several research fields; see~\cite{Krizhevsky-2012, Goodfellow-2013, Wan-2013} for examples. One nice feature of neural networks and other machine learning approaches is that they alleviate the computational burden at the operation stage---by shifting computationally intensive `hard work' to the off-line training stage. 

However, accurately learning the end-to-end mapping from the available measurements to the exact network state is very challenging in our context---the accuracy achieved by convergent Gauss-Newton iterates ({under good initializations}) is hard to obtain using learning approaches. The mapping itself is very complex, necessitating a wide and/or deep NN that is hard to train with reasonable amounts of data. ~\ahmed{In addition, training a deep NN (DNN) is computationally cumbersome requiring significant computing resources.} Also, DNN slows down real-time estimation, as passing the input through its layers is a sequential process that cannot be parallelized. To circumvent these obstacles, we instead propose to train a \textit{shallow} neural network to `learn to initialize'---that is, map the available measurements to  a point in the neighborhood of the true latent states, which is then used to initialize Gauss-Newton; see Fig.~\ref{fig:idea} for an illustration. 
\ahmed{When the Gauss-Newton solver is initialized at a point in the vicinity of the optimal solution, it enjoys quadratic convergence~\cite{ferreira2009local}; otherwise, divergence is possible.}
We show that such a hybrid machine learning / optimization approach yields superior performance compared to conventional optimization-only approaches, in terms of stability, accuracy, and runtime efficiency. We demonstrate these benefits using convincing experiments with the benchmark IEEE-37 distribution feeder with several renewable energy sources installed and several types of phasor and conventional measurements, as well as pseudo-measurements. The key to success is \textit{appropriate design of the NN training cost function} for the `neighborhood-finding' NN. As we will see, the proposed cost function serves our purpose much better than using a generic cost function for conventional NN training. \ahmed{In addition, owing to the special design of the training cost function, the experiments corroborate the resiliency of the proposed approach in case of modest network reconfiguration events.}

  \begin{figure}[htbp]
  	\centering
  	\includegraphics[width=0.5\columnwidth]{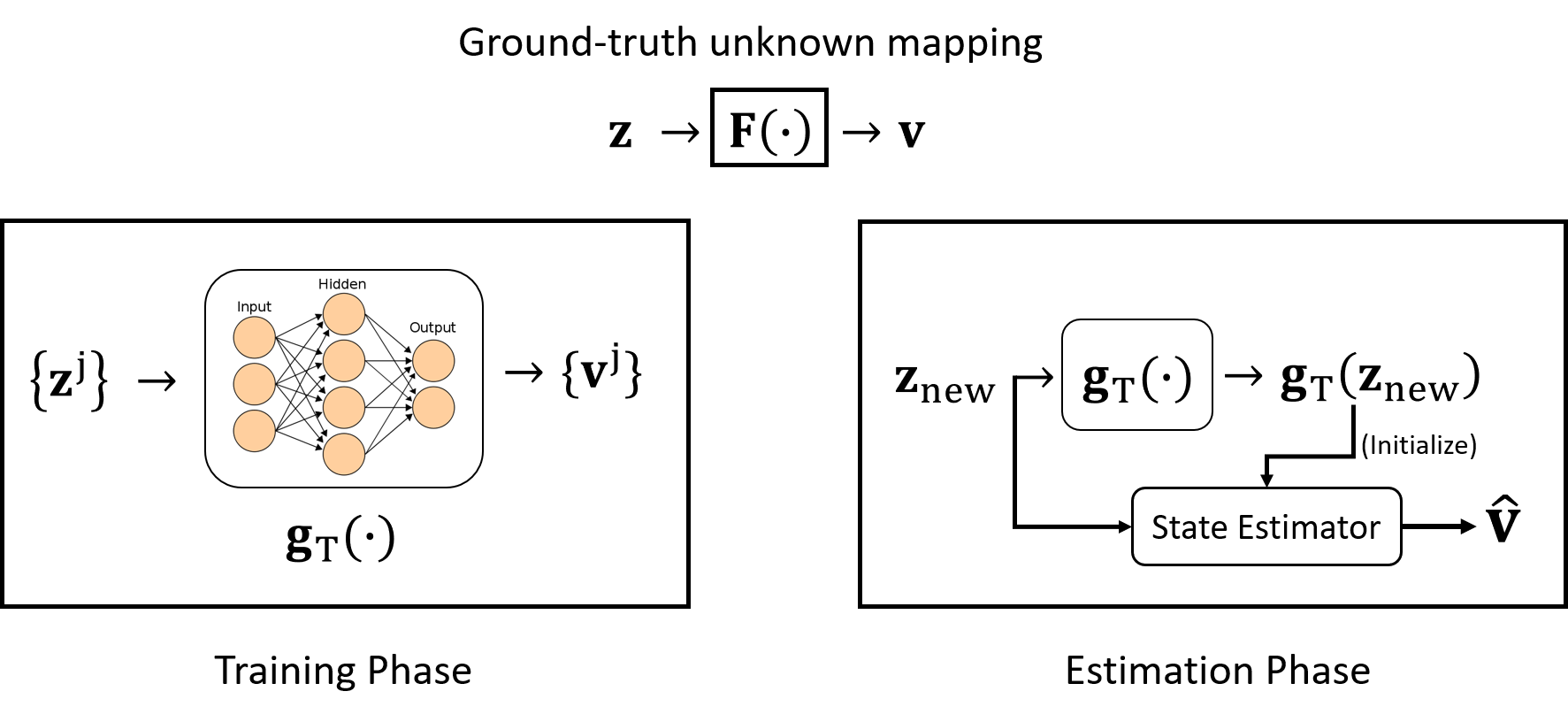}
  	\caption{The proposed learning-based DSSE}
  	\label{fig:idea}
  \end{figure}
  
\noindent {\bf Context.} Machine learning approaches are not entirely new in the power systems / smart grid area. For instance, an online learning algorithm was used in~\cite{o2010residential} to shape residential energy demand and reduce operational costs. In~\cite{fang2011online}, a multi-armed bandit online learning technique was employed to forecast the power injection of renewable energy sources. An early example of using NNs in estimation problems appeared in~\cite{amin1997system} as part of damage-adaptive intelligent flight control. {Closer to our present context,\cite{manitsas2012distribution} proposed the use of an artificial neural network that takes the measured power flows as input and aims to estimate the bus injections, which are later used as pseudo-measurements in the state estimation. In contrast to our approach where the NN is used to approximate the network state given the conventional measurements as well as the pseudo-measurements, the authors of~\cite{manitsas2012distribution} designed an artificial NN to generate pseudo-measurements from the available power flow measurements.} 
\ahmed{In addition, artificial neural networks have been used for the prediction step in dynamic state estimation~\cite{rousseaux1988dynamic}, and for forecasting-aided state estimation~\cite{da1993state, Coutto-2009} where the state of the network is estimated from the previous sequence of states, without using conventional measurements to anchor the solution. This is better suited for transmissions systems, which are more predictable relative to distribution systems with time-varying loads.} {As the installation of renewable energy sources in the distribution grid surges, the huge volatility brought by these energy sources~\cite{29825} induces rapid changes in network state, and thus the previously estimated state is often bad initialization for the next instance of DSSE.}
To the best of our knowledge, machine learning approaches have not yet been applied to the core DSSE optimization task, which is the focus of our work.

	
\noindent{\bf Notation}: matrices (vectors) are denoted by boldface capital (small) letters; $(\cdot)^T$, $ \overline{(\cdot)} $ and $(\cdot)^H$ stand for transpose, complex-conjugate and complex-conjugate transpose, respectively; and $|(\cdot)|$ denotes the magnitude of a number or the cardinality of a set.
	
\section{Distribution System State Estimation}
\subsection{Network Representation}
Consider a multi-phase distribution network consisting of $ N + 1 $ nodes and $ L $ edges represented by a graph $ \mathcal{G} := (\mathcal{N},\mathcal{L}) $, whose set of multi-phase nodes (buses) is {indexed by} $ \mathcal{N} := \{0, 1, \ldots, N\} $, and $ \mathcal{L} \subseteq \mathcal{N} \times \mathcal{N} $ represents the lines in the network. Let the node $ 0 $ be the substation that connects the system to the transmission grid. The set of phases at bus $ n $ and line $ (l, m) $ are denoted by $ \boldsymbol{\varphi}_{n} $ and $ \boldsymbol{\varphi}_{lm} $, respectively. Let the voltage at the $ n $-th bus for phase $ \phi $ be denoted by $ v_{n,\phi} $. Then, define $ {\bf v}_{n} := [v_{n, \phi}]_{\phi\in \boldsymbol{\varphi}_n}$ to collect the voltage phasors at the phases of bus $ n $. In addition, let the vector $ {\bf v} $ concatenate the vectors $ {\bf v}_n $ for all $ n \in \mathcal{N} $. Lines $ (l, m) \in \mathcal{L}$ are modeled as $ \pi $-equivalent circuit, where phase impedance and shunt admittance are denoted by $ {\bf Z}_{lm} \in \mathbb{C}^{|\boldsymbol{\varphi}_{lm}|\times|\boldsymbol{\varphi}_{lm}|} $ and $ \check{\bf Y}_{lm} \in \mathbb{C}^{|\boldsymbol{\varphi}_{lm}|\times|\boldsymbol{\varphi}_{lm}|} $, respectively.
	
	\subsection{Problem Formulation}
	The DSSE task amounts to recovering the voltage phasors of buses given measurements related to real-time physical quantities, and the available pseudo-measurements. Actual measurements are acquired by smart meters, PMUs, and $ \mu $PMUs that are placed at some locations in the distribution network. The measured quantities are usually noisy and adhere to
	\begin{equation}\label{eq:meas1}
	\tilde{z}_\ell = \tilde{h}_\ell({\bf v}) + \xi_\ell, \qquad 1 \leq \ell \leq L_m
	\end{equation}
	where $ \xi_\ell $ accounts for the zero-mean measurement noise with known variance $ \tilde{\sigma}_\ell^2 $. The functions $ \tilde{h}_\ell({\bf v}) $ are dependent on the type of the measurement, and can be either linear or quadratic relationships. In the next section, the specific form of $ \tilde{h}_\ell({\bf v}) $ will be discussed. In addition, load and generation forecasting methods are employed to obtain pseudo-measurements that can help with identifying the network state. The forecasted quantities are modeled as 
	\begin{equation}\label{eq:meas2}
	\check{z}_\ell = \check{h}_\ell ({\bf v}) + \zeta_\ell, \qquad 1 \leq \ell \leq L_s
	\end{equation}
	where $ \zeta_\ell $ represents the zero-mean forecast error which has a variance of $ \check{\sigma}_\ell^2 $. Since $ \check{z}_\ell ${'s} represent power-related quantities, they are usually modeled as quadratic functions of the state variable $ {\bf v} $. While the value of the measurement noise variance {$ \tilde{\sigma}_\ell^2 $} 
	depends on the accuracy of the measuring equipment, the variance of the forecast error can be determined using historical forecast data.
	
	Let $ {\bf z} $ be a vector of length $ L = L_m + L_s $ containing the measurements and pseudo-measurements, and $ {\bf h}({\bf v}) $ the equation relating the measurements to the state vector $ {\bf v} $, which will be specified in the next section.
	Adopting a weighted least-squares formulation, the problem can be cast as follows
	\begin{align} \label{eq:psse}\nonumber
	\min_{\bf v}\ J({\bf v}) &= \sum_{\ell = 1}^{L_m} \tilde{w}_\ell \big(\tilde{z}_\ell - \tilde{h}_\ell ({\bf v})\big)^2 + \sum_{\ell = 1}^{L_s} \check{w}_\ell \big(\check{z}_\ell - \check{h}_\ell ({\bf v})\big)^2\\
	&= \ ({\bf z} - {\bf h}({\bf v}))^T {\bf W} ({\bf z} - {\bf h}({\bf v}))
	\end{align}
	where the values of $ \tilde{w}_\ell $ and $ \check{w}_\ell $ are inversely proportional to $ \sigma_\ell^2 $ and $ \check{\sigma}_\ell^2 $, respectively. The optimization problem~\eqref{eq:psse} is {non-convex} due to the nonlinearity of the measurement mappings $ {\bf h}({\bf v}) $ inside the squares.

	
%
	
	\subsection{Available Measurements for DSSE}
As indicated in the previous subsection, only few real-time measurements are usually available in distribution networks, relative to the obtainable measurements in transmission systems. Therefore, pseudo-measurements are used to alleviate the issue of solving an underdetermined problem. 
{There are always different latencies for different sources of measurements which bring up the issue of time skewness. Many approaches have been proposed in the literature to tackle the issue~\cite{Zhang-2013}. In this work, assume that the issue is resolved using one of the solutions in the literature, and hence, the measurements are assumed to be synchronized.}
First, the measurements function $ \tilde{h}({\bf v}) $ will be introduced for all types of available measurements. Then, the construction of the pseudo-measurements mappings $ \check{h}({\bf v}) $ will be explained. 
	
The measurement functions consist of:

	\noindent
	$\bullet$ {\it phasor measurements} which represent the complex nodal voltages $ {\bf v}_n $, or current flows $ {\bf i}_{lm} $. The corresponding measurement function is linear in the state variable $ {\bf v} $. These measurements are usually obtained by the PMUs and $ \mu $PMUs. Each measurement of this type is handled as two measurements, i.e., the real and imaginary parts of the complex quantities are handled as two measurements. For the nodal voltages, the real and imaginary parts are given as follows
	\begin{equation}
	\Re\{v_{n,\phi}\} =\ \frac{1}{2}\ {\bf e}_{n,\phi}^T\ ({\bf v}_n + \overline{\bf v}_n),
	\end{equation}
	\begin{equation}
	\Im\{v_{n,\phi}\} =\ \frac{1}{2j}\ {\bf e}_{n,\phi}^T\ ({\bf v}_n -  \overline{\bf v}_n)
	\end{equation}
	where $ {\bf e}_{\phi} $ is the $ \phi $-th canonical basis in $ \mathbb{R}^{|\boldsymbol{\varphi}_{n}|} $. In addition, the current flow measurements can be modeled as
	\begin{equation}
	\Re\{i_{lm,\phi}\} = \frac{1}{2}\ {\bf e}_{lm,\phi}^T\ \big( {\bf Y}_{lm} ({\bf v}_l - {\bf v}_m) +  \overline{\bf Y}_{lm} (\overline{\bf v}_l  - \overline{\bf v}_m ) \big)
	\end{equation}
	\begin{equation}
	\Im\{i_{lm,\phi}\} = \frac{1}{2j}\ {\bf e}_{lm,\phi}^T\ \big( {\bf Y}_{lm} ({\bf v}_l - {\bf v}_m) -  \overline{\bf Y}_{lm}  (\overline{\bf v}_l  - \overline{\bf v}_m ) \big)
	\end{equation}
	where $ {\bf Y}_{lm} $ is the inverse of $ {\bf Z}_{lm} $, and $ {\bf e}_{lm,\phi} $ is the $ \phi $-th canonical basis in $ \mathbb{R}^{|\boldsymbol{\varphi}_{lm}|} $.
	
	\noindent
	$\bullet$ {\it real-valued measurements} which encompass voltage magnitudes $ |{v}_{n,\phi}| $, current magnitudes $ |i_{lm,\phi}| $, and real and reactive power flow measurements $ p_{lm,\phi}, q_{lm,\phi} $. These measurements are obtained by SCADA systems, Distribution Automation, Intelligent Electronic Devices, and PMUs. The real-valued measurements are nonlinearly related to the state variable $ {\bf v} $. The measured voltage magnitude square, and active and reactive power flows can be represented as quadratic functions of the state variable $ {\bf v} $, see~\cite{DallAnese13}. The current magnitude squared can be written as follows
	\begin{equation}
	|i_{lm,\phi}|^2 = ({\bf v}_l - {\bf v}_m)^H {\bf y}_{lm,\phi}^H {\bf y}_{lm,\phi} ({\bf v}_l - {\bf v}_m)
	\end{equation}
	where $ {\bf y}_{lm,\phi} $ is the $ \phi $-th row of the admittance matrix $ {\bf Y}_{lm} $. Therefore, all the real-valued measurements can be written as quadratic measurements of the state variable $ {\bf v} $.

	{The available real-time measurements are usually insufficient to `{pin down}' the network state, {as we have discussed}. In this case, the system is said to be unobservable. 
	Hence, pseudo-measurements that augment the real-time measurements are crucial in DSSE as they help achieve network observability. Pseudo-measurements are obtained through load and generation forecast procedures that aim at estimating the energy consumption or generation utilizing historical data and location-based information. They are considered less accurate than real-time measurements, and hence, assigned low weights in the WLS formulation.} The functions governing the mapping from the state variable to the forecasted load and renewable energy source injections can be formulated as quadratic functions~\cite{DallAnese13, Zamzam-2016}. 
	
	Therefore, any measurement synthesizing function $ h_\ell({\bf v}) $ can be written in the following form
	\begin{equation}
	h_\ell ({\bf v}) = {\overline{\bf v} }^T {\bf D}_\ell {\bf v} + {\bf c}_\ell^T {\bf v} + {\overline{\bf c}_\ell }^T \overline{\bf v} 
	\end{equation}
	where $ {\bf D}_\ell $ is a Hermitian matrix. This renders $ J({\bf v}) $ a fourth order function of the state variable, {which is} very challenging to optimize 

	The Gauss-Newton algorithm linearizes the first order optimality conditions to iteratively update the state variables until convergence. The algorithm is known to perform well in practice given that the algorithm is initialized from a point in the vicinity of the true network state, albeit lacking provable convergence {result in theory}. Several variants of the algorithm have been proposed in the literature using polar~\cite{baran1994}, rectangular~\cite{nuqui2007hybrid} and complex~\cite{dzafic2018hybrid} representations of the state variables. All these algorithms work to a certain extent, but failure cases are also often observed. {Again,} stable convergence performance is only observed when the initialization is close enough to the optimal solution of~\eqref{eq:psse}. This is not entirely surprising---given the non-convex nature of the DSSE problem.
	
	
	\section{Proposed Approach: Learning-aided DSSE Optimization}
	{Assume that there exists a mapping ${\bf F}(\cdot)$ such that $${\bf F}({\bf z})={\bf v};$$ i.e., ${\bf F}(\cdot)$ maps the (noiseless) measurements to the ground-truth states. An example of such mapping is an optimization algorithm that can optimally solve the DSSE problem in the noiseless case, assuming that the solution is unique. The algorithm takes ${\bf z}$ as input and outputs ${\bf v}$. In reality the actual (and the virtual) measurements will be noisy, so we can only aim for  $${\bf F}({\bf z}) \approx {\bf v};$$  which is also what optimization-based DSSE aims for in the noisy case. 
		
Inspired by the recent successes of machine learning, it is intriguing to ask whether it is possible to learn mapping ${\bf F}(\cdot)$  from historical data. If the answer is affirmative and the learned $\hat{\bf F}(\cdot)$ is easy to evaluate, then the DSSE problem could be solvable in a very efficient way {\em online}, after the mapping $\hat{\bf F}(\cdot)$  is learned {\em offline}.  

In machine learning, neural networks are known as universal function approximators. In principle, a three-layer (input, hidden, output) NN can approximate any continuous multivariate function down to prescribed accuracy, if there are no constraints on the number of neurons ~\cite{cybenko1989approximation}. This motivates us to consider employing an NN for approximating ${\bf F}({\bf z})$ in the DSSE problem. An NN with vector input ${\bf z}$, vector output $ {\bf g} $, and one hidden layer comprising $T$ neurons synthesizes a function of the folowing form  
\begin{equation}\label{eq:NN} 
	{\bf g}_T({\bf z}) = \sum_{t=1}^{T} \boldsymbol{\alpha}_t \sigma({\bf w}_t^T {\bf z} + \beta_t), 
\end{equation}
where ${\bf w}_t$ represents the linear combination of the inputs in ${\bf z}$ that is fed to the $t$th neuron (i.e., the unit represented by $\sigma({\bf w}_t^T {\bf z} + \beta_t)$), ${\beta}_t$ the corresponding \ahmed{scalar} bias, and the \ahmed{vectors} $\boldsymbol{\alpha}_t$'s combine the outputs of the neurons in the hidden layer to produce the vector output of the NN.  
	The parameters $(\boldsymbol{\alpha}_t, {\bf w}_t, {\beta}_t)_{t=1}^T$ can be learned by minimizing the training cost function 
	\begin{equation}\label{eq:existing}
	  \min_{\{\boldsymbol{\alpha}_t,{\bf w}_t,{\beta}_t\}_{t=1}^T}~\sum_{j}\| {\bf v}^j - {\bf g}_T({\bf z}^j) \|_2^2, 
	\end{equation}
	where the pair $ ({\bf z}^j, {\bf v}^j) $ is a training sample of measurements and the corresponding underlying voltages to be estimated, in our context. 

	The above training cost function ideally seeks an NN that works perfectly---at least over the training set. 
	This approach is similar in spirit to the one in~\cite{sun2017learning}, which considered a problem in wireless resource allocation with the objective of `learning to optimize'---meaning, training an NN to learn the exact end-to-end input-output mapping of an optimization algorithm. Our experience has been that, for DSSE, such an approach works to some extent, but its  performance is not ideal. Trying to learn the end-to-end DSSE mapping appears to be too ambitious, requiring very large $T$ or a deep NN, and very high training sample complexity.  To circumvent these obstacles, we instead propose to train a \textit{shallow} neural network, as above, to `learn to initialize'---that is, map the available measurements to  a point in the neighborhood of the true latent state, which is then used to initialize Gauss-Newton \ahmed{as depicted in Fig.~\ref{fig:Arch}}. 
	
			\begin{figure}[h]
		\centering
		\includegraphics[width=0.5\columnwidth]{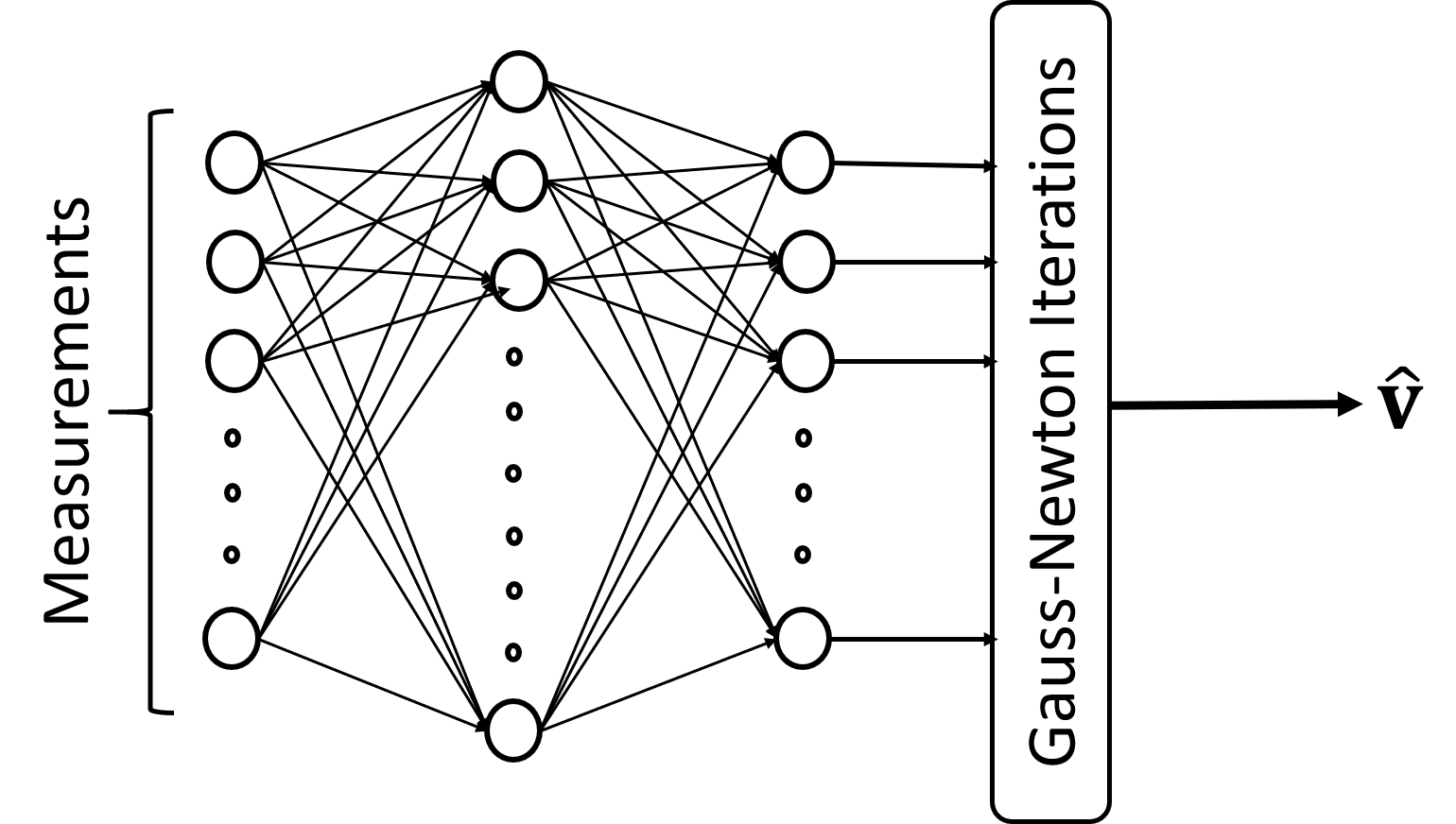}
		\caption{Learning-based state estimator structure}
		\label{fig:Arch}
	\end{figure}
	
	More specifically, we propose using the following cost function for training the NN:
	\begin{equation}\label{eq:elf2}
	 \min_{\{{\bf w}_t,{\beta}_t,\boldsymbol{\alpha}_t\}_{t=1}^T} \sum_{j}\ \max\{ \lVert {\bf v}^j - {\bf g}_T ({\bf z}^j) \rVert_2^2 - \epsilon^2,\ 0\}
	\end{equation}
	where the cost function indicates that the NN parameters are tuned with the relaxed goal that ${\bf g}_T({\bf z}^j) $ lies in the ball of radius $ \epsilon $ around ${\bf v}^j$. Fig.~\ref{fig:cost} illustrates the effect of changing the value of $ \epsilon $ on the empirical loss function. The high-level idea is as follows: instead of enforcing minimization of $\sum_{j}\| {\bf v}^j - {\bf g}_T({\bf z}^j) \|_2^2$, we seek a `lazy' solution such that $\lVert {\bf v}^j - {\bf g}_T ({\bf z}^j) \rVert_2^2 \leq \epsilon$ for as many $j$ as possible---in other words, it is enough to get to the right neighborhood. As we will show, this `lowering of the bar' can significantly reduce the complexity of the NN (measured by the number of neurons) that is needed to learn such an approximate mapping, and with it also the number of training samples required for learning. These complexity benefits are obtained while still reproducing a point that is close enough to serve as a good initialization for Gauss-Newton, ensuring stable and rapid convergence. To back up this intuition, we have the following result. 
	\begin{proposition}\label{prob:one}
		Let $ \sigma(\cdot) $ be any continuous sigmoidal function, and let ${\bf g}_T({\bf z}): \mathbb{R}^L\rightarrow \mathbb{R}^K$ be in the form $$ {\bf g}_T({\bf z}) = \sum_{t=1}^{T} \boldsymbol{\alpha}_t \sigma({\bf w}_t^T {\bf z} + \beta_t). $$ Then, for approximating a continuous mapping ${\bf F}: \mathbb{R}^L\rightarrow \mathbb{R}^K$, 
		the complexity for a shallow network to solve Problem~\eqref{eq:elf2} exactly (i.e., with zero cost) {for a finite number of bounded training samples} $\big({\bf z}^j,{\bf v}^j={\bf F}({\bf z}^j)\big)$  is at least in the order of $$ T = \mathcal{O}\bigg(\Big({\frac{\epsilon}{\sqrt{K}}}\Big)^{-\frac{L}{r}}\bigg). $$ where $r$ is the number of continuous derivatives of ${\bf F}(\cdot)$.  
	\end{proposition}
	The proof of this proposition is 
	relegated to Appendix A. {Note that the boundedness assumption on the inputs is a proper assumption since these quantities represent voltages and powers}. The implication here is very interesting, as controlling $\epsilon$ can drastically reduce the required $T$ (and, along with it, sample complexity) while still ensuring an accurate enough prediction to enable rapid convergence of the ensuing Gauss-Newton stage. Furthermore, keeping the network shallow and $T$ moderate makes the actual online computation (passing the input measurements through the NN to obtain the sought initialization) simple enough for real-time operation. This way, the relative strengths of learning-based and optimization-based methods can be effectively combined, and the difficulties of both methods can be circumvented.

	One important remark is that Proposition~\ref{prob:one} is derived under the assumption that ${\bf F}(\cdot)$ 
	is a continuous mapping that can be parametrized with $L$ parameters, which is hard to verify in our case. Nevertheless, we find that the theoretical result here is interesting enough and intuitively pleasing. In a case of a simple single-phase feeder, the state estimation mapping is indeed continuous and finitely parametrizable; see Appendix B.
	More importantly, as will be seen, this corroborating theory is consistent with our empirical results.

	\begin{figure}[h]
		\centering
		\includegraphics[width=0.45\columnwidth]{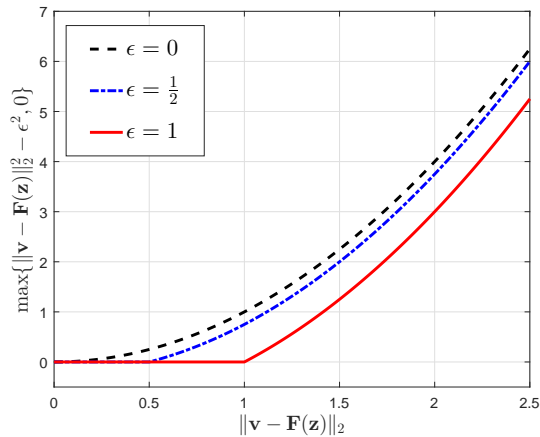}
		\caption{The empirical loss function used for training}
		\label{fig:cost}
	\end{figure}

In order to tune the neural network parameters, $ N_t $ training samples have to be used in order to minimize the cost function in~\eqref{eq:elf2}. Two different ways can be utilized in order to obtain such training data. {First, historical data for load and generation can be utilized. Note that these data are not readily available unless all the buses in the network are equipped with measuring devices, however, such load and generation profiles can be estimated using a state estimation algorithm. Then, the network power flow equations can be solved to obtain the system state which is used later to synthesize the measurements using~\eqref{eq:meas1} and~\eqref{eq:meas2}.} Hence, for each historical load and generation instance, a {noiseless} training pair $ ({\bf z}^j, {\bf v}^j) $ can be generated. The second way to obtain the training data is to resort to an operating state estimation procedure. In this case, the goal of the neural network approach is to emulate the mapping of the estimator from the measurements space to the state space. The second approach suffers all the limitations of the current state estimation algorithms such as inaccuracy or computational inefficiency. In addition to providing noisy training pairs, these limitations result in a much more time consuming way of generating training data. Therefore, the first way is adopted for the rest of this paper, and the detailed procedure is presented in the experiments section.
\ahmed{	
\begin{remark*}
	One concern for data-driven approaches is that the mapping is learned from historical data under a certain network topology. What if the configuration changes for some reason (e.g., maintenance)? Is the trained mapping still useful? The answer is, surprisingly, affirmative, thanks to the `lazy' training objective---since we do not seek exact solutions for the mapping, modest reconfiguration of the network will not destroy the effectiveness of the trained NN for initialization, as we will see in Section V.
\end{remark*}
}

	\section{Experimental Results}

	The proposed state estimation procedure is tested on the benchmark IEEE-37 distribution feeder. \ahmed{This network is recommended for testing state estimation algorithms by the Test Feeder Working Group of the Distribution System Analysis Subcommittee of the IEEE PES~\cite{Feeder-WG}.} The feeder is known to be a highly unbalanced system that has several delta-connected loads, which are blue-colored in Fig.~\ref{fig:Network}. 
	\begin{figure}[htbp]
		\centering
		\includegraphics[width=0.5\columnwidth]{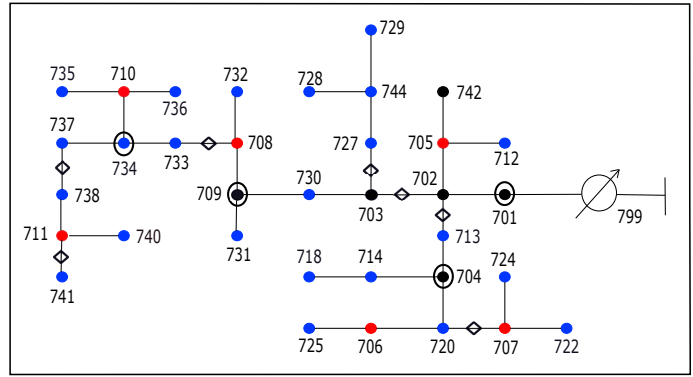}
		\caption{IEEE-37 distribution feeder. Nodes in blue are with loads, and red nodes represent buses with DER installed. Buses with PMUs are circled, and the links where the current magnitudes are measured have a small rhombus on them.}
		\label{fig:Network}
	\end{figure}

	The feeder has nodes that feature different types of connections, i.e., single-, two-, and three-phase connections. Additionally, distributed energy resources are assumed to be installed at six different buses, which are colored in red in Fig.~\ref{fig:Network}. In Table~\ref{tab:connections}, the types of the connections of all the loads and DERs are presented where (L) and (G) mean load and DER, respectively.
	\begin{table}[htbp]
		\centering
		\caption{Loads and DER Connections.}
		\begin{tabular}{|c|c|c||c|c|c|}
			\hline
			{\textbf{Bus}} & {\textbf{Type}} &\textbf{Connections} &\textbf{Bus} & {\textbf{Type}} &{\textbf{Connections}}\\
			\hline
			705     & (G)  & a-b, b-c  &  728 &  (L) & a-b, b-c, c-a\\
			\hline
			706     & (G)  & b-c       &  729 &  (L) & a-b\\
			\hline
			707     & (G)  & b-c, c-a  &  730 &  (L) & c-a\\
			\hline
			708     & (G)  & b-c       &  731 &  (L) & b-c \\
			\hline
			710     & (G)  & a-b       &  732 &  (L) & c-a\\
			\hline
			711     & (G)  & c-a       &  733 &  (L) & a-b\\
			\hline
			712     & (L)  & c-a       &  734 &  (L) & c-a\\
			\hline 
			713     & (L)  & c-a       &  735 &  (L) & c-a\\
			\hline
			714     & (L)  & a-b, b-c  &  736 &  (L) & b-c\\
			\hline
			718     & (L)  & a-b       &  737 &  (L) & a-b\\
			\hline
			720     & (L)  & c-a       &  738 &  (L) & a-b\\
			\hline
			722     & (L)  & b-c, c-a  &  740 &  (L) & c-a \\
			\hline
			724     & (L)  & b-c       &  741 &  (L) & c-a\\
			\hline
			725     & (L)  & b-c       &  742 &  (L) & a-b, b-c\\
			\hline
			727     & (L)  & c-a       &  744 &  (L) & a-b\\
			\hline 
		\end{tabular}%
		\label{tab:connections}%
	\end{table}%
	
	Historical load and generation data available in~\cite{bank2013development} modulated by the values of the loads are used to generate the training samples. Each time instance has an injection profile which is used as an input to the linearized power flow solver in~\cite{Garces2015}. The algorithm returns a voltage profile (network state variable) which is utilized to generate the value of the measurements at this point of time. A total of $ 100,000 $ loading and generation scenarios were used to train a shallow neural network. {The network has an input size of $ 103 $, $ 2048 $ nodes in the hidden layer, and output of size $ 210 $.} 
	
	The available measurements are detailed as follows.
	\begin{itemize}
		\item \emph{PMU measurements}: four PMUs are installed at buses $ 701 $, $ 704 $, $ 709 $, and $ 734 $ which are circled in Fig.~\ref{fig:Network}. It assumed that the voltage phasors of all the phases are measured at these buses. This sums up to $ 12 $ complex measurement, i.e., $ 24 $ real measurements. {We installed a unit at the substation, and then placed the rest to be almost evenly distributed along the network in order to achieve better observability.}
		\item \emph{Current magnitude measurements}: The magnitude of the current flow is measured on all phases of the lines that are marked with a rhombus in Fig.~\ref{fig:Network}. The number of current magnitude measurements is $ 21 $ real measurements. {We installed the units such that the state estimation problem can be solved without unobservability problems. We tested different installation for the current flow measuring devices with noiseless measurements, and then chose one such that the problem is not ill-posed.}
		\item \emph{Pseudo-measurments}: The aggregate load demand of the buses with load installed, which are blue-colored in Fig.~\ref{fig:Network}, are estimated using a load forecasting algorithm using historical and situational data. Therefore, only two real quantities are obtained by the state estimator that relate to the active and reactive estimated load demand at the load buses. In addition, an energy forecast method is used to obtain an estimated injection from the renewable energy sources located at the DER buses which are colored in red in Fig.~\ref{fig:Network}. The total number of load buses in the feeder is $ 23 $, and the number of distributed energy sources is $ 6 $. Therefore, the state estimator obtains $ 58 $ real pseudo-measurements relating to the active and reactive forecasted demand/injection at these buses.
	\end{itemize}

	The state estimator obtains noisy measurements and inexact load demands and energy generation quantities. It is assumed that the noise in the PMU voltage measurements is drawn from a Gaussian distribution with zero mean and a standard deviation of $ 10^{-3} $. Additionally, the noise added to current magnitudes is Gaussian distributed with a standard deviation of $ 10^{-2} $. Finally, the differences between the pseudo-measurement and the real load demand and generations are assumed to be drawn from a Gaussian distribution with a standard deviation of $ 10^{-1} $. \ahmed{The proposed learning-based state estimation approach aims at estimating the voltage phasor at all the phases of all the buses in the network.}

	The shallow neural network is trained using the TensorFlow~\cite{tensorflow2015-whitepaper} software library with $ 90 \%$ of the data used for training while the rest is used for verification. After tuning the network parameters, noisy measurements are generated and then passed to the state estimator architecture in Fig.~\ref{fig:Arch}.
	In order to show the effect of the modified cost function, we test the networks trained with different values of $ \epsilon $ on $ 1,000 $ loading and generation scenarios. 
	Fig.~\ref{fig:hist} shows the histogram of the distance between the output of the shallow NN and the true network state. With the conventional training cost function ($ \epsilon = 0 $) the resulting distribution is more spread than the histogram that we obtain through the network trained with a relaxed cost function ($ \epsilon = 1 $). 

	\begin{figure}[t]
		\centering
		\includegraphics[width=0.45\columnwidth]{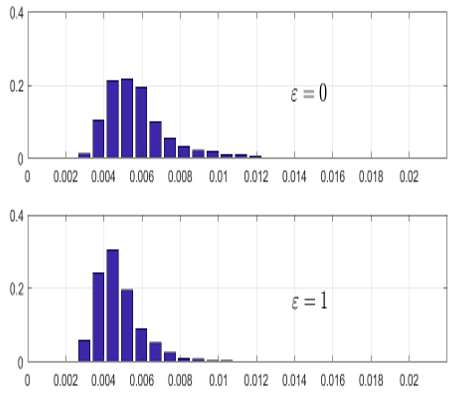}
		\caption{Histogram of the distance between the shallow NN output and the true voltage profile with ($ \epsilon = 0 $) and ($ \epsilon = 1 $).}
		\label{fig:hist}
	\end{figure}
	
	Two performance indices~\eqref{eq:estimation-accuracy}-\eqref{eq:cost} are introduced to quantify the quality of the estimate as well as the performance of the proposed approach. The first index, which is denoted by $ \nu $, represents the Frobenius norm square of the estimation error. Also, the cost function value at the estimate is denoted by $ \mu $. 
	
	\begin{equation}\label{eq:estimation-accuracy}
	\nu = \| \hat{\bf v} - {\bf v}^{\text{true}}\|_2^2
	\end{equation}
	\begin{equation}\label{eq:cost}
	\mu = \sum_{\ell = 1}^{L} (z_{\ell} - h_\ell(\hat{\bf v}))^2
	\end{equation}
	
	Furthermore, in order to show the effect of changing the cost function used for training, the average cost achieved using the proposed approach is shown in Table~\ref{tab:eps} when different values of $ \epsilon $ are used for training cost function. In addition, the average number of iterations required by the Gauss-Newton iterates to converge to the optimal estimate is also presented. Using a positive value of $ \epsilon $ can lead to savings up $ 25\% $ in computations, which is valuable when solving the DSSE for large systems. Also, it can be seen that choosing non-zero values for $ \epsilon $ enhances the performance of the proposed architecture. The estimation accuracy can be almost $ 5 $ times better using a positive $ \epsilon $. As the approximation requirement is relaxed while training the shallow NN, the network gains in generalization ability, accommodating more scenarios of loading and generation profiles. 
	
	
		\begin{table}[htbp]
			\centering
			\caption{The estimator performance with different values of ($ \epsilon $).}
			\begin{tabular}{c c c}
				\hline
				\hline\rule{0pt}{3ex} 
				{$ \boldsymbol{\epsilon} $} & {\# Iterations} &{$ \mu $}\\
				\hline\rule{0pt}{3ex} 
				{$ \bf 0 $}     &  $ 7.035 $ & $ 8.968\times 10^{-3} $ \\
				\rule{0pt}{3ex}
				{$ \bf \frac{1}{8} $}     &  $ 6.825 $ & $ 5.531\times 10^{-3}  $ \\
				\rule{0pt}{3ex} 
				{$ \bf \frac{1}{4} $}     &  $ 6.095 $ & $ 3.417\times 10^{-3}  $ \\
				\rule{0pt}{3ex} 
				{$ \bf \frac{1}{2} $}     &  $ 5.675 $ & $ 1.822\times 10^{-3}  $ \\
				\rule{0pt}{3ex} 
				{$ \bf \frac{1}{\sqrt{2}} $}     &  $ 5.220 $ & $ 5.056\times 10^{-3}  $ \\
				\rule{0pt}{3ex} 
				{$ \bf 1 $}     &  $ 6.150 $ & $ 5.859\times 10^{-3}  $ \\
				\rule{0pt}{3ex} 
				{$ \bf 2 $}     &  $ 6.415 $ & $ 1.365\times 10^{-2}  $ \\
				\hline 				\hline 
			\end{tabular}%
			\label{tab:eps}%
		\end{table}%
		
	
	To assess the efficacy of the proposed approach we compare it against the complex variable Gauss-Newton state estimator using~\cite{sorber2012unconstrained} as a state-of-art Gauss-Newton solver for a real-valued optimization problem in complex variables. The shallow NN was trained with $ \epsilon = \frac{1}{2} $ in the next comparisons. 
	
	
	
		\begin{table}[htbp]
			\centering
			\caption{Performance comparison of different state estimators}
			\begin{tabular}{c c c}
				\hline
				\hline\rule{0pt}{3ex} 
				{\textbf{Method}} & {{$\nu$}} &{$ \mu $}\\
				\hline\rule{0pt}{3ex} 
				{\bf Proposed}     &  $ 9.558\times 10^{-3} $ & $ 1.822\times 10^{-3}  $ \\
				\rule{0pt}{3ex} 
				{\bf G-N}     & $ 9.845\times 10^{-2} $  & $ 4.861 \times 10^{-2}$     \\
				\hline 				\hline 
			\end{tabular}%
			\label{tab:accuracy}%
		\end{table}%
	
	In Table ~\ref{tab:accuracy}, the average accuracy achieved in estimating the true voltage profile using both the Gauss-Newton method and the proposed architecture is presented for $ 1000 $ scenarios. In the Gauss-Newton implementation, {the complex voltages provided by the PMUs are used to initialize the voltage phasors corresponding to these buses. This provides a better initialization point to the Gauss-Newton algorithm which also enhances its stability.}
	Still, the proposed approach is able to achieve almost $ 10 $ times better accuracy on average. In addition, the fitting error which represents the WLS cost function is greatly enhanced using the proposed approach. 
	
		\begin{table}[htbp]
		\centering
		\caption{Timing and convergence of different state estimators}
			\begin{tabular}{c c c}
				\hline
				\hline\rule{0pt}{3ex} 
				{\textbf{Method}} & {{Time (ms)}} &{\# Divergence}\\
				\hline\rule{0pt}{3ex} 
				{\bf Proposed}     &  $ 347 $ & $ 0 $ \\
				\rule{0pt}{3ex} 
				{\bf G-N}     & $ 2468 $  & $ 28 $     \\
				\hline 				\hline 
			\end{tabular}%
			\label{tab:timing}%
		\end{table}%
	
	In order to assess the computational time of the proposed algorithm, we tried $ 1000 $ simulated cases for the NN-{assisted} state estimator and the Gauss-Newton ({optimization-only}) state estimator. In Table~\ref{tab:timing}, the number of divergent cases out of the $ 1000 $ trials is presented for both approaches. While the Gauss-Newton approach failed to converge in $ 28 $ scenarios, the proposed architecture has converged for all considered cases. In addition, the time taken by the proposed learning approach is almost four times less than the Gauss-Newton algorithm. This is due to the fact that only few Gauss-Newton iterations need to be done when the proposed approach is utilized.
	\ahmed{
	\section{System Reconfiguration}
	In distribution systems, the network configuration may be subject to changes either for restoration~\cite{Morelato-89}, i.e., to isolate a fault, or for system loss reduction~\cite{baran1989network,Gomes-2006}. An important task is to identify the underlying topology of the feeder. In order to perform this task, several approaches have been recently developed utilizing measurement data~\cite{Deka-2018,Cavraro2018GraphAF,Weng2017DistributedER}. Without access to the correct network topology, accurate state estimation is untenable, as the estimator will attempt to fit the measurements to a wrong model. In other words, the ground truth function generating $ {\bf z} $ is different from $ {\bf h}({\bf v}) $, and hence, solving~\eqref{eq:psse} is meaningless. Therefore, in this study we only consider the case where the (new) system topology has been adequately identified. Hence, in the latter part of the proposed approach, the Gauss-Newton iterations utilize accurate system topology information.
	
	In order to test the robustness of the proposed approach, we assume that  switches are available on several lines in the feeder as depicted in Fig.~\ref{fig:Network}. Also, we add three additional lines to the network as redundant lines that are assumed to be unenergized under normal operating conditions~\cite{DallAnese-Reconfig}. Specifically, switches are assumed to be present on the original lines $ (710, 735)$, $ (703, 730) $ and $ (727, 7444) $, and $ 3 $ additional tie lines $ (742, 744) $, $ (735, 737) $ and $ (703, 741) $ are added to the feeder. We asses the robustness of the proposed learning approach under the following three scenarios.
	\begin{itemize}
		\item Scenario A: a fault has occurred in line $ (727, 744) $ and the tie-switch on line $ (742, 744) $ has been turned on;
		\item Scenario B: a fault has occurred in line $ (703, 730) $ and the line $ (703, 741) $ has been energized; and
		\item Scenario C: the switch on line $ (710, 735) $ has been turned off, and the switch on line $ (735, 737) $ has been connected.
	\end{itemize}

	We train the NN using data that are generated from the original network topology. We assess the performance of the proposed learning-based state estimator on Scenarios A, B and C. Note that, although in these scenarios the neural network is trained on a different generating model, the proposed approach still has advantageous performance when compared with the plain Gauss-Newton approach. This can be attributed to the specially designed cost function~\eqref{eq:elf2} that was proposed to train the NN. During the course of our experiments, we noticed that the robust performance against topology reconfiguration events is more pronounced when positive $ \varepsilon $ is used for training the NN. Table \ref{tab:reconfig} compares the performance of the proposed state estimator with $ (\varepsilon = \frac{1}{2}) $ against the plain Gauss-Newton approach under the three system reconfiguration events. Clearly, the proposed approach still provides performance gains even under modest topology changes. {When the approach was tested under significant reconfiguration events, our simulations showed that the initialization produced by projecting the flat voltage profile onto the linear space defined by the PMU measurements performs better than the neural network initialization. The training procedure is not computationally intensive due to the simplicity of the model and the relaxed training cost function. Therefore, in case of severe system reconfigurations that are expected to last for long time, the shallow neural network can be retrained in the order of a few minutes to match the underlying physical model.}
	
	\begin{table}[htbp]
		\centering
		\caption{Performance comparison with system reconfiguration.\\(Averaged over 100 runs)}\label{tab:reconfig}
		\begin{tabular}{c| c |c c c}
			\hline
			\hline 
			\multicolumn{2}{c|}{\textbf{Scenario}} & {{Time (ms)}} &{$ \nu $} & $ \mu $\\
			\hline\rule{0pt}{3ex} 
			\multirow{2}{*}{\bf A}     & \bf{Proposed} & $ 476 $ & $ 3.894 \times 10^{-2} $ & $ 3.172 \times 10^{-2}$\\\cline{2-5}
			\rule{0pt}{3ex} 
			& \bf{G-N}& $ 1574 $ & $ 7.460 \times 10^{-2} $ & $ 1.147 \times 10^{-1} $\\\hline\hline
			\rule{0pt}{3ex} 
			\multirow{2}{*}{\bf B}     & \bf{Proposed}& $ 2339 $  & $ 7.963 \times 10^{-2} $ &  $ 1.051 \times 10^{-2} $  \\\cline{2-5}
			\rule{0pt}{3ex} 
			& \bf{G-N}& $ 3696 $ & $ 9.207 \times 10^{-2} $ & $ 2.050 \times 10^{-2} $ \\\hline\hline
			\rule{0pt}{3ex} 
			\multirow{2}{*}{\bf C}     & \bf{Proposed}& $ 429 $  & $ 1.297 \times 10^{-2} $  & $ 3.150 \times 10^{-3}$  \\\cline{2-5}
			\rule{0pt}{3ex} 
			& \bf{G-N}& $ 1568 $ & $ 8.527 \times 10^{-2} $ & $ 1.289 \times 10^{-1} $ \\
			\hline 				\hline 
		\end{tabular}%
	\end{table}%
	}

	\section{Conclusion}
	This paper presented a data-driven learning-based state estimation architecture for distribution networks. The proposed approach designs a neural network that can accommodate several types of measurements as well as pseudo-measurements. Historical load and energy generation data is used to train a neural network in order to produce an approximation of the network state. Then, this estimate is fed to a Gauss-Newton algorithm for refinement. Our realistic experiments suggest that the combination offers fast and reliable convergence to the optimal solution. The IEEE-37 test feeder was used to test the proposed approach in scenarios that include distributed energy sources. The proposed learning approach shows superior performance results in terms of the accuracy of the estimates as well as computation time.

\section{Appendix A\\Proof of Proposition 1}
To prove the proposition, we first invoke the following lemma:
\begin{Lemma}[{\cite[Theorem~$2 $]{cybenko1989approximation}}~] Let $ \sigma(\cdot) $ be any continuous sigmoidal function. Then, given any function $ f(\cdot) $ that is continuous on the $ d $-dimensional unit cube $ {\bm I}_d = [0, 1]^d $, and $ {\epsilon} > 0 $, there is a sum, $ g(\cdot):  \mathbb{R}^d \rightarrow \mathbb{R} $, of the form 
	\begin{equation}
	g({\bf z}) = \sum_{t=1}^{T} \alpha_t \sigma({\bf w}_t^T {\bf z} + \beta_t) 
	\end{equation} 
	for which, 
	\begin{equation}\notag
	| g({\bf z}) - f({\bf z})| < {\epsilon} \qquad \forall {\bf z} \in {\bm I}_d.
	\end{equation}
\end{Lemma}

\begin{proof}[Proof of Proposition~1]
	Note that the vector-valued function $ {\bf F}(\cdot) $ can be represented as $ K $ separate scalar-valued functions. 
	In order to prove the proposition, we start by considering approximating a scalar-valued function $ f_k({\bf z}) $ that represents the mapping between $ {\bf z} $ and the $ k $-th element of $ {\bf F}({\bf z}) $. 
	
	Since $ {\bf z}^j $'s are finite with length $ L $, finite maximum and minimum along each dimension can be obtained. Let the vectors that collect the maximum and minimum values be denoted by $ \overline{\bf z} $ and $ \underline{\bf z} $, respectively. Then, each training sample $ {\bf z}^j $ is replaced by $ \tilde{\bf z}^j  = {\bf D}_{\overline{\bf z} - \underline{\bf z}} ({\bf z}^j - \underline{\bf z}) $, where $ {\bf D}_{\overline{\bf z} - \underline{\bf z}} $ is a diagonal matrix that has the values of $ \overline{\bf z} - \underline{\bf z} $ on the diagonal. Therefore, the vectors $ \tilde{\bf z}^j $ are inside the $ L $-dimensional cube $ {\bm I}_L $. According to Lemma~1, there exists a sum $ \tilde{g}_k(\tilde{\bf z}) $ in the form of
	\begin{equation}\label{eq:single_out_NN}
	\tilde{g}_k(\tilde{\bf z}) = \sum_{t=1}^{T_k} \tilde{\alpha}_{t,k} \sigma(\tilde{{\bf w}}_{t,k}^T {\bf z} + \tilde{\beta}_{t,k}) 
	\end{equation} 
	that satisfies 
	\begin{equation}
	| f_k(\tilde{\bf z}^j) - \tilde{g}_k(\tilde{\bf z}^j) | < {\epsilon}_1 \qquad \forall\ \tilde{\bf z}^j
	\end{equation}
	for $ {\epsilon}_1 > 0 $. 
	Then, let $ {g}_k({\bf z}) $ be a mapping in the form of~\eqref{eq:single_out_NN} where the parameters are given by 
	\begin{eqnarray}\notag
	\alpha_{t,k}\! = \tilde{\alpha}_{t,k},\ \ \beta_{t,k}\! =\tilde{\beta}_{t,k}\! -\! \tilde{\bf w}_{t,k}^T {\bf D}_{\overline{\bf z} - \underline{\bf z}}\ \underline{\bf z},\ \ {\bf w}_{t,k}\! = {\bf D}_{\overline{\bf z} - \underline{\bf z}} \tilde{\bf w}_{t,k}.
	\end{eqnarray}
	Then, for all $ {\bf z}^j $ we have
	\begin{equation}
	| f_k({\bf z}^j) - {g}_k({\bf z}^j) | < \epsilon_1.
	\end{equation}
	
	This result holds for each of the $K$ scalar elements of ${\bf F}({\bf z})$. Therefore, by parallel concatenation of the $K$ neural networks used to approximate the $K$ scalar-valued functions, we obtain a shallow neural network that has $ K $ outputs and ($ \sum_i T_i $) neurons at the hidden layer. Setting $ \epsilon = \sqrt{K}\ \epsilon_1 $, we deduce that there exists a sum $ {\bf g}_T({\bf z}) $ in the form of~\eqref{eq:NN} that satisfies
	\begin{equation}
	\| {\bf F}({\bf z}^j) - {\bf g}_T({\bf z}^j) \|_2 < {\epsilon} \qquad \forall\ {\bf z}^j.
	\end{equation}
	It is clear now that the parameters of this function $ {\bf g}_T({\bf z}) $, i.e., $ \boldsymbol{\alpha}_t $, $ {\bf w}_t $, and $ \beta_t $, achieve a zero cost function solving Problem~\eqref{eq:elf2}, and hence is optimal in solving~\eqref{eq:elf2}.

	In addition, since an approximation can be realized using any sigmoid functions, the main result in~\cite{mhaskar1996neural} specifies that the minimum number of neurons required to achieve accuracy at least $ \epsilon_1 $ for a scalar-valued function is given by 
	\begin{equation}
	T = \mathcal{O}(\epsilon_1^{-\frac{L}{r}})
	\end{equation}
	where $ r $  denotes the number of continuous derivatives of the approximated function $ f({\bf z}) $, and $ L $ represents the number of parameters of the function. {In order to achieve $ \epsilon $ accuracy for approximating $ {\bf F}({\bf z}) $, at least one of the real-valued functions that construct $ {\bf F}({\bf z}) $ has to achieve $ \frac{\epsilon}{\sqrt{K}} $.}
	Hence, the complexity of shallow neural networks that optimally solve~\eqref{eq:elf2} for $\epsilon > 0 $ is at least
	\begin{equation}\notag
	T = \mathcal{O}\bigg(\Big({\frac{\epsilon}{\sqrt{K}}}\Big)^{-\frac{L}{r}}\bigg).
	\end{equation}
\end{proof}

\section*{Appendix B\\ An Example Network}
The neural networks are known as universal functions approximators. Nevertheless, the theoretical results on the ability to approximate function are usually limited to continuous functions. Hence, for continuous functions, the neural networks are expected to be able to achieve high approximation accuracy. 
Unfortunately, checking the continuity of the state estimation solution, which is an inverse mapping of a highly nonlinear function, is not simple to be checked. 

\begin{figure}[t]
	\centering
	\includegraphics[width=0.45\columnwidth]{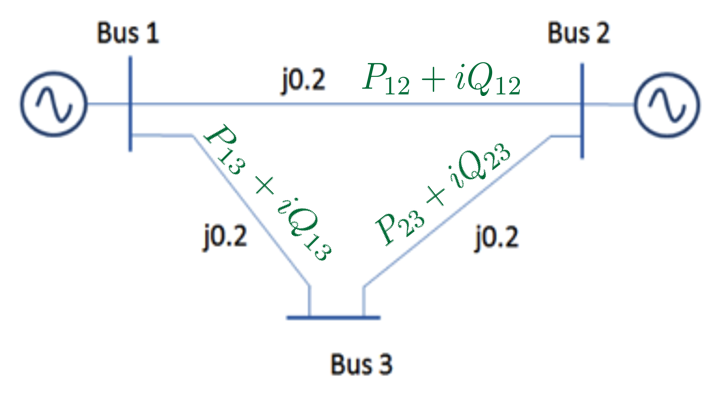}
	\caption{Example $ 3 $-bus network}
	\label{fig:example}
	\vspace{0pt}
\end{figure}

In this appendix, a $ 3 $-bus balanced lossless network is presented, in Fig.~\ref{fig:example}, in order to inspect the continuity of the state estimation mapping. We assume that the simple network has $ 3 $ buses and that the magnitude of the voltages are measured at all buses. In addition, the active and reactive power flows are measured at all lines. Since, the phase at Bus $ 1 $ can be taken as a reference for the other buses, the state estimation problem amounts to estimating the lines phase differences, or equivalently, the phases at Bus $ 2 $ and Bus $ 3 $.

The power flow equations can be expressed as follows
\begin{eqnarray}
P_{12} =& B_{12}|v_1| |v_2| \sin(\theta_{12}),\\
Q_{12} =& |v_1|^2 - B_{12}|v_1| |v_2| \cos(\theta_{12}),\\
P_{13} =& B_{12} |v_1| |v_3| \sin(\theta_{13}),\\
Q_{13} =& |v_1|^2 - B_{12} |v_1| |v_3| \cos(\theta_{13})
\end{eqnarray}
where $ B_{ij} $ is the susceptance of the line between Bus $ i $ and Bus $ j $, $ |v_i| $ is the voltage magnitude at the $ i $-th Bus, and $ \theta_{ij} $ is the angle difference on the line $ (i,j) $. Assuming that the collected measurements are noiseless, the solution of the state estimation problem can be written in closed-form as
\begin{equation}\label{eq:map1}
\theta_{12} = \sin^{-1} \Big(\frac{P_{12}}{B_{12}|v_1| |v_2|}\Big),
\end{equation}
\begin{equation}\label{eq:map2}
\theta_{13} = \sin^{-1} \Big(\frac{P_{13}}{B_{13}|v_1| |v_3|}\Big).
\end{equation}
\begin{claim}
	The mapping between the measurements $ P_{12}, P_{13}, |v_1|, \text{ and } |v_2| $ and the state of the network, i.e., $ \theta_{12} $ and $ \theta_{13} $, is continuous if 
	\begin{equation}
	B_{12}|v_1| |v_2| \geq \epsilon,\text{ and }\qquad B_{13}|v_1| |v_3| \geq \epsilon
	\end{equation}
	for any $ \epsilon >0 $.
\end{claim}
\begin{proof}
	The proof is straightforward and build upon basic results from real functions analysis. First, the function
	\begin{equation}
	f_1 ( P_{12}, |v_1|, |v_2| ) = \frac{P_{12}}{B_{12}|v_1| |v_2|}
	\end{equation}
	is continuous on $ B_{12}|v_1| |v_2| \in [\epsilon, \infty] $ for any $ \epsilon > 0  $. Then, the mapping functions in~\eqref{eq:map1} is composite function of $ f_1 $ and $ \sin^{-1}(\cdot) $ which is a continuous function. Therefore, the mapping in~\eqref{eq:map1} is continuous on $ B_{12}|v_1| |v_2| \in [\epsilon, \infty] $ for any $ \epsilon > 0  $~\cite{rudin1964principles}. The same follows for $ \theta_{13} $.
\end{proof}

	\bibliographystyle{IEEEtran}
	\bibliography{RefPSSE}

\end{document}